\documentclass[final,3p,times]{elsarticle}
\usepackage{amsfonts}
\usepackage{amssymb}
\usepackage{amsmath}
\usepackage{amsthm}
\usepackage{graphics}
\usepackage{graphicx}
\usepackage{xcolor}
\usepackage[colorlinks = true, linkcolor = blue, 
urlcolor = blue, citecolor = blue]{hyperref}
\newtheorem{theorem}{Theorem}

\newtheorem{definition}{Definition}

\newtheorem{proposition}{Proposition}
\newtheorem{observation}{Observation}

\makeatletter
\def\ps@pprintTitle{   
\let\@oddhead\@empty   
\let\@evenhead\@empty   
\let\@oddfoot\@empty   
\let\@evenfoot\@oddfoot
}
\makeatother

\begin{document}
\begin{frontmatter}
\title{Construction of noisy bound entangled states and the range criterion}

\author{Saronath Halder}
\ead{saronath.halder@gmail.com}

\author{Ritabrata Sengupta}
\ead{rb@iiserbpr.ac.in}

\address{Department of Mathematical Sciences, 
Indian Institute of Science Education and Research 
Berhampur,\\ Transit Campus,  Government ITI, Berhampur 
760010, Odisha, India}

\begin{abstract}
In this work we consider bipartite noisy bound entangled states with 
positive partial transpose, that is, such a state can be written as a convex 
combination of an edge state and a separable state. In particular, we 
present schemes to construct distinct classes of noisy bound entangled 
states which satisfy the range criterion. As a consequence of the present 
study we also identify noisy bound entangled states which do not satisfy 
the range criterion. All of the present states are constituted by exploring 
different types of product bases. 
\end{abstract}

\begin{keyword}
Bound entanglement, Positive partial transpose, Edge state, Range criterion, 
Unextendible product basis,  Uncompletable product basis
\end{keyword}
\end{frontmatter}

\section{Introduction}\label{sec1}
One of the key developments within the theory of quantum entanglement 
\cite{Horodecki09, Gune09}, is the invention of bound entangled states 
\cite{Horodecki97}. These states are mixed entangled states from which 
entanglement in pure form cannot be extracted by local operations and 
classical communication \cite{Horodecki98}. This holds true even if large 
number of identical copies of the state are shared among spatially separated 
parties. Since the discovery of bound entangled states, there is no simple 
technique to identify such states. Therefore, it is highly nontrivial to present 
new classes of bound entangled states. For a given bipartite quantum state 
if the state produces negative eigenvalue(s) under partial transpose then 
it guarantees inseparability of that state \cite{Peres96}. The problem 
arises when the given state remains positive under partial transpose (PPT). 
In such a situation it is not always easy to conclude whether the state is 
separable or inseparable (entangled). Generally, for an arbitrary bipartite 
PPT state if the dimension of the corresponding Hilbert space is greater 
than 6 then it is difficult to say whether the state is separable or inseparable 
\cite{Horodecki96}. In fact, the problem of determining any density matrix 
-- separable or entangled is a NP-hard problem \cite{Gurvits04}. However, 
if a PPT state is entangled then the state must be bound entangled 
\cite{Horodecki98}. On the other hand existence of bound entangled states 
with negative partial transpose is conjectured and remains open till date 
\cite{DiVincenzo00, Dur00, Pankowski10}.  

Application of the {\it range criterion} is quite effective approach to prove 
the inseparability a given PPT state \cite{Horodecki97}. For a given bipartite 
density matrix $\rho$, if the state is separable then there exists a set of 
product states \{$|\theta_i\rangle_1\otimes|\theta_i\rangle_2$\} that spans 
the range of $\rho$ while the set of product states \{$|\theta_i\rangle_1
\otimes|\theta_i^\ast\rangle_2$\} spans the range of $\rho^t$. Here, the 
superscript $t$ denotes the partial transpose operation (considering second 
subsystem) and $\ast$ denotes the complex conjugation in a basis with 
respect to which the partial transpose is taken. Any state which violates the 
range criterion is an entangled state. However, there exist several classes 
of PPT entangled states which satisfy the range criterion 
\cite{Bandyopadhyay05, Bandyopadhyay08}. Evidently, detection of such 
states are one of the troublesome tasks in the entanglement theory. 
Therefore, to understand these states in a better way, it is important to 
constitute such states. Note that a full-rank state trivially satisfy the range 
criterion. So, it is significant to understand the forms of distinct classes of 
low-rank bound entangled states which satisfy the range criterion.

An efficient scheme to produce bound entangled states is related to 
unextendible product bases (UPBs). In Ref.~\cite{Bennett99}, which 
introduces UPB, it was shown that for a given Hilbert space $\mathcal{H}$ 
if the states within a UPB span the subspace $\mathcal{H}_S$ of 
$\mathcal{H}$ then the normalized projector onto the complementary 
subspace $\mathcal{H}_S^\perp$ is a PPT entangled state. The bipartite 
bound entangled states produced in this manner are edge states and they 
violate the range criterion in an extreme way. This is because an edge state 
$\rho$ has a property that there exists no product state $|\theta_1\rangle
\otimes|\theta_2\rangle$ in its range such that $|\theta_1\rangle\otimes|
\theta_2^\ast\rangle$ belongs to the range of $\rho^t$ \cite{Lewenstein01}. 
So, it is interesting to explore the states which are not edge states still 
violet the range criterion. In this context, it is important to mention about 
uncompletable product bases (UCPBs). For a given Hilbert space, a UCPB 
cannot be extended to a complete orthogonal product basis 
\cite{DiVincenzo03}. Furthermore, in the same paper it was shown that for 
a given Hilbert space $\mathcal{H}$ if the states within a UCPB span the 
subspace $\mathcal{H}_S$ of $\mathcal{H}$ then the normalized projector 
onto the complementary subspace $\mathcal{H}_S^\perp$ may or may not 
be a PPT entangled state.

For a practical purpose, it is difficult to say how to use an arbitrary bound 
entangled state. Nevertheless, in last few years use of several classes of bound 
entangled states were discussed in different contexts, for example, secure key 
distillation \cite{Horodecki05, Horodecki08, Horodecki09-1}, quantum metrology 
\cite{Toth18} etc. Bipartite bound entangled states which are related to quantum 
steering \cite{Moroder14} and quantum nonlocality \cite{Vertesi14} were also 
explored. Later, a family of nonlocal bound entangled states were constructed 
in Ref.~\cite{Yu17}. In the present work we consider bipartite {\it noisy} bound 
entangled states with positive partial transpose (see Ref.~\cite{Sindici18} as well). 
Any of the present states can be produced by mixing a separable state (noise) 
with an edge state \cite{Lewenstein01}. Therefore, these bound entangled states 
are useful to learn about the robustness of entanglement within an edge state 
(also go through Ref.~\cite{Bandyopadhyay08} in this context). In 
Ref.~\cite{Toth18} the authors showed some examples of bipartite bound 
entangled states, entanglement of which is robust against noise. Again, it was 
shown that the bound entangled states within which entanglement is robust 
against noise are fit for experimental verification \cite{Sentis18}. Clearly, the 
study of noisy bound entangled states has got practical importance.

We now give the main findings of the present work: (i) Starting from a 
particular class of UPBs, we show how to construct bipartite noisy bound 
entangled states that satisfy the range criterion. The range of such a state 
is spanned by a set of orthogonal product states. In particular, we obtain 
that these bound entangled states have lower rank with respect to the 
states presented in \cite{Bandyopadhyay05} for a given Hilbert space. (ii) 
Next, we give a protocol to constitute bipartite noisy bound entangled 
states from a particular class of UCPBs. An important property of these 
bound entangled states is that they satisfy the range criterion though the 
range of such a state cannot be spanned by a set of orthogonal product 
states. (iii) We further explore the construction of other classes of bipartite 
noisy bound entangled states. A subset of which do not satisfy the range 
criterion.

Rest of the paper is arranged in the following way: In Sec.~\ref{sec2}, we 
give few definitions and preliminary ideas that are helpful to describe the 
present constructions. Next, in Sec.~\ref{sec3}, the main results of this paper 
are presented. Finally, in Sec.~\ref{sec4}, the conclusion is drawn.

\section{Preliminaries}\label{sec2}
In this section we first give the definitions of different types of product 
bases namely UPB and UCPB for bipartite quantum systems. For more 
general definitions, one can go through the Refs.~\cite{Bennett99, 
DiVincenzo03}. Furthermore, we discuss about the bound entangled states 
produced from UPBs and UCPBs. We also discuss about some existing tools 
to examine the inseparability of a given bipartite PPT state.

\begin{definition}\label{def1}
Let $\mathcal{H}$ = $\mathcal{H}_A\otimes\mathcal{H}_B$ be a bipartite 
quantum system. Consider a set $S$ of pure orthogonal product states which 
span a subspace $\mathcal{H}_S$ of $\mathcal{H}$. Now, the states of $S$ 
form a \emph{UCPB} if the complementary subspace $\mathcal{H}_S^\perp$ 
contains fewer pure orthogonal product states than its dimension. On the other 
hand, the states of $S$ form a \emph{UPB} if the complementary subspace 
$\mathcal{H}_S^\perp$ contains no product state.
\end{definition}

Note that the vectors which span the subspace $\mathcal{H}_S^\perp$ are 
all orthogonal to the vectors in $\mathcal{H}_S$. Clearly, a UCPB cannot be 
extended to a full basis for a given Hilbert space $\mathcal{H}$. This is 
because if the states of the set $S$ form a UCPB then these states along 
with few other mutually orthogonal product states in $\mathcal{H}_S^\perp$ 
are not sufficient to cover the whole dimension of $\mathcal{H}$. Moreover, 
if the states of the set $S$ form a UPB then it is not possible to find any 
product state which is orthogonal to all the states of $S$. We now consider 
a set of bipartite product states $\{|\phi_i\rangle = |\alpha_i\rangle\otimes
|\beta_i\rangle\}_{i=1}^n$  where $|\phi_i\rangle\in\mathcal{H}$ = 
$\mathcal{H}_A\otimes\mathcal{H}_B$ for each $i$. Also consider that this 
set forms an unextendible product basis which spans the subspace 
$\mathcal{H}_S$ of $\mathcal{H}$. So, the normalized projector onto 
$\mathcal{H}_S^\perp$ can be written as

\begin{equation}\label{eq1}
\rho = \frac{1}{D-n}\left(I-\sum_{i=1}^n|\phi_i\rangle\langle\phi_i|\right),
\end{equation}
where $D$ is the total dimension of the composite quantum system 
$\mathcal{H}$ and $I$ is the identity operator acting on $\mathcal{H}$. The 
density matrix $\rho$ is a bipartite bound entangled state \cite{Bennett99}. 
Note that if the states $|\phi_i\rangle$ are the normalized vectors of real vector 
space then the state $\rho$ must be invariant under partial transpose. Next, 
consider a set $S^\prime$ of pure orthogonal product states $\{|\phi_i\rangle 
= |\alpha_i\rangle\otimes|\beta_i\rangle\}_{i=1}^{n^\prime}$. Also assume 
that these states form a UCPB in a Hilbert space $\mathcal{H}$ of dimension 
$D$. So, if the states $\{|\phi_i\rangle\}_{i=1}^{n^\prime}$ span the 
subspace $\mathcal{H}_{S^\prime}$ of $\mathcal{H}$ then $\mathcal{H}_{S
^\prime}^\perp$ contains the product states $\{|\phi_i\rangle = |\alpha_i
\rangle\otimes|\beta_i\rangle\}_{i = n^\prime+1}^{n}$, where $n<D$. The 
normalized projector onto the subspace $\mathcal{H}_{S^\prime}^\perp$ is 
given by

\begin{equation}\label{eq2}
\rho^\prime = \frac{1}{D-n^\prime}\left(I-\sum_{i=1}^{n^\prime}
|\phi_i\rangle\langle\phi_i|\right) 
= \frac{D-n}{D-n^\prime}\left(\frac{1}{D-n}\left(I-\sum_{i=1}^{n}
|\phi_i\rangle\langle\phi_i|\right)\right)
+ \frac{n-n^\prime}{D-n^\prime}\left(\frac{1}{n-n^\prime}\sum_{i=n
^\prime +1}^{n}|\phi_i\rangle\langle\phi_i|\right),
\end{equation}
where $I$ is the identity operator acting on the same Hilbert space where the 
states of $S^\prime$ belong. From the above it is obvious that if the state 
$\rho^\prime$ is entangled then it must be a noisy bound entangled state. In 
this sense, UCPBs have an important role to produce noisy bound entangled 
states. Moreover, such a state is partial transpose invariant if the states of the 
given UCPB are normalized vectors of a real vector space. In this work we mostly 
discuss about the bound entangled states which are partial transpose invariant. 
Remember that to prove a partial transpose invariant state satisfies the range 
criterion, it is sufficient to find a set of pure product states (that are normalized 
vectors of a real vector space) which spans the range of the given state. 

To examine whether a given PPT state is inseparable or not, indecomposable 
positive (P) maps which are not completely positive (CP) play an important role. 
For example, the celebrated Choi map \cite{Choi75} is one such map. There is 
a rich literature on constructions, examples, and applications of such maps (for 
instance see Ref.~\cite{Chruscinski14} and the references therein). One can 
also use suitable witness operators \cite{Chruscinski14, Lewenstein01} to do the 
above. Here, we consider only PPT states and thus, to prove the separability or 
inseparability of those states, indecomposable P maps which are not CP or the 
witness operators are quite relevant. However, we start by giving the definition 
of witness operators.

\begin{definition}\label{def2}
An entanglement witness $\mathcal{W}$ is a Hermitian operator with the 
properties that (a) Tr($\mathcal{W}\delta$) $\geq$ 0, for any separable 
density matrix $\delta$ and (b) there is at least an inseparable density 
matrix $\rho$ for which Tr($\mathcal{W}\rho$) $<$ 0. 
\end{definition}

Note that these witness operators can be taken in normalized form, i.e., 
Tr($\mathcal{W}$) = 1. The following witness operator, we are going to 
use extensively in our paper. This witness operator is efficient enough to 
detect any UPB generated bound entangled state. Suppose, $\rho$ is a 
UPB generated bound entangled state as given in Eq.~(\ref{eq1}). To 
detect this state we further consider the entanglement witness operator 
\cite{Lewenstein01, Bandyopadhyay05, Terhal01}, given by

\begin{equation}\label{eq3}
\mathcal{W} = \Pi -\gamma I,
\end{equation}
where $\Pi$ is a projector onto the subspace spanned by the states 
within the UPB and $I$ is the identity operator acting on the same 
Hilbert space where the states of the UPB belong. The parameter 
$\gamma$ can be defined in the following way:

\begin{equation}\label{eq4}
\gamma = \mbox{min}\langle\phi|\Pi|\phi\rangle,
\end{equation}
where the minimization is taken over all separable states $|\phi\rangle$, 
belong to the Hilbert space where the states of the UPB belong. By 
construction of $\mathcal{W}$, we obtain 

\begin{equation}\label{eq5}
\mbox{Tr}\big(\mathcal{W}\rho\big) = -\gamma < 0.
\end{equation}

Therefore, $\mathcal{W}$ witnesses the state $\rho$. Notice that the 
structure of the operator ensures the fact that the trace of the above 
equation must be $\geq0$ if the state $\rho$ is a separable state. In this 
context, it is important to mention that entanglement witness operators 
for bipartite states can be constructed from a positive but not completely 
positive map. Let us consider a P but not CP map $\Lambda$, where 
$\Lambda: M_d\rightarrow M_d$. The operator $(I\otimes\Lambda)
|\Psi\rangle\langle\Psi|$ can be used to witness entanglement of some 
states, here $I$ is a $d\times d$ identity matrix and $|\Psi\rangle$ is a 
maximally entangled state in $d \otimes d$. 

For a given entangled state, there always exists a positive but not 
completely positive map such that the map detects the entanglement 
of the state \cite{Horodecki96}. Taking inner automorphism of this map, 
it is possible to detect other entangled states that are locally equivalent 
to the given state. This particular fact can be realized in the following 
way: Let $\rho$ be an entangled state which is detected by a positive 
but not completely positive map $\Lambda$, i.e., $\exists$ $|\psi\rangle$ 
such that $\langle\psi|(I\otimes\Lambda)\rho|\psi\rangle < 0$. Now, 
consider any operator $\rho^\prime$ = $(A\otimes B)\rho(A\otimes B)
^\dagger$; $A$ and $B$ are invertible operators. Rewriting the state 
$\rho$ as $\sum_{i,j}|i\rangle\langle j|\otimes\rho_{ij}$ in the block 
matrix form, we get 

\begin{equation}\label{eq6}
\rho^\prime = (A\otimes B)\left(\sum_{i,j}|i\rangle\langle j|\otimes
\rho_{ij}\right)(A\otimes B)^\dagger = \sum_{i,j}(A|i\rangle\langle 
j|A^\dagger)\otimes (B\rho_{ij}B^\dagger).
\end{equation}
Applying the map $\Lambda$ on one of the subsystems of the operator 
$\rho^\prime$, we obtain

\begin{equation}\label{eq7}
(I\otimes\Lambda)\rho^\prime = \sum_{i,j}(A|i\rangle\langle j|
A^\dagger)\otimes\Lambda(B\rho_{ij}B^\dagger) = (A\otimes I)
\left(\sum_{i,j}|i\rangle\langle j|\otimes\Lambda (B\rho_{ij}B^\dagger)
\right)(A^\dagger\otimes I). 
\end{equation}
Notice that whether the above is {\it positive} or {\it negative}, solely 
depends on the term $\sum_{i,j}|i\rangle\langle j|\otimes\Lambda 
(B\rho_{ij}B^\dagger)$. We further consider a different map $\Lambda
^\prime$ and apply it instead of $\Lambda$, where the action of 
$\Lambda^\prime$ can be defined as $\Lambda^\prime(X)$ = 
$\Lambda(B^{-1}XB^{-1})$. Finally, consider the following 

\begin{equation}\label{eq8}
\langle\psi|(I\otimes\Lambda^\prime)(I\otimes B)\rho(I\otimes B)^\dagger
|\psi\rangle = \langle\psi|\sum_{i,j}|i\rangle\langle j|\otimes\Lambda^\prime
(B\rho_{ij}B^\dagger)|\psi\rangle = \langle\psi|\sum_{i,j}|i\rangle\langle j|
\otimes\Lambda(\rho_{ij})|\psi\rangle = \langle\psi|(I\otimes\Lambda)
\rho|\psi\rangle<0.
\end{equation}
Thus, application of local invertible operators basically help to detect some 
extra entangled states with the same P but not CP map (in this context see 
also Refs.~\cite{Guhne06}). However, for better understanding of this 
technique we consider few known PPT entangled states which satisfy the range 
criterion and discuss their detection using the Choi map along with a local 
unitary operator (given in the next section). This also helps us to understand 
the main results presented in the next section as applying the same technique 
we prove the inseparability of few PPT states in $d\otimes 3$. Note that to 
prove the inseparability of the other PPT states we use the witness operator 
given in Eq.~(\ref{eq3}).

\section{Main results}\label{sec3}
In Ref.~\cite{Sengupta13}, it is shown that along with a suitable unitary 
it is possible to use Choi map to detect the PPT entangled state which is 
generated from a $3\otimes3$ UPB, given in Ref.~\cite{Bennett99}. The 
UPB and corresponding PPT entangled state are given by 

\begin{equation}\label{eq9}
\begin{array}{c}
|\psi_1\rangle = \frac{1}{\sqrt{2}}|1\rangle|1-2\rangle,~~ 
|\psi_2\rangle = \frac{1}{\sqrt{2}}|1-2\rangle|3\rangle,~~ 
|\psi_3\rangle = \frac{1}{\sqrt{2}}|3\rangle|2-3\rangle,~~
|\psi_4\rangle = \frac{1}{\sqrt{2}}|2-3\rangle|1\rangle,~~\\
|\psi_5\rangle = \frac{1}{3}|1+2+3\rangle|1+2+3\rangle,~~

\rho = \frac{1}{4}\left(I-\sum_{i=1}^5
|\psi_i\rangle\langle\psi_i|\right).
\end{array}
\end{equation}
Here the notation $|a\pm b \pm c\rangle$ denotes the vector $|a\rangle\pm
|b\rangle\pm|c\rangle$. We use this notation throughout the paper. Note that 
the Choi map alone cannot detect the PPT entangled state $\rho$ of 
Eq.~(\ref{eq9}). The action of Choi map with a unitary operator is given in 
the following equation:

\begin{eqnarray}\label{eq10}
\Lambda: M_3 \rightarrow M_3,~~
(I\otimes\Lambda_u)\rho = (I\otimes\Lambda)(I\otimes u)\rho(I\otimes u)
^\dagger, \nonumber \\
\Lambda: ((a_{ij})) \rightarrow 
\frac{1}{2}
\left[\begin{array}{ccc}
a_{11} + a_{22} &        -a_{12}       &        -a_{13}       \\
       -a_{21}       &  a_{22} + a_{33} &        -a_{23}       \\
    -a_{31}       &        -a_{32}       &  a_{33} + a_{11} \\
\end{array}\right],
\end{eqnarray}
where $u$ is a unitary operator. Here, we apply with the following unitary 
operator: 

\begin{equation}\label{eq11}
u =\left[
\begin{array}{ccc}
 \frac{1}{2}  &      \frac{\sqrt{3}}{2}        &            0         \\ 
-\frac{\sqrt{3}}{2}  &      \frac{1}{2}        &            0         \\ 
    0        &          0             &            1         \\
\end{array}\right].
\end{equation}

\begin{observation}\label{obs1}
It is sufficient to apply the Choi map along with the unitary $u$ to detect 
certain bound entangled states in $3\otimes3$ which satisfy the range 
criterion. 
\end{observation}

We consider two distinct classes of PPT states $\rho_1(\lambda)$, 
$\rho_2(\lambda)$. Originally, the values of $\lambda$ for which these 
states are entangled, can be found in Refs.~\cite{Bandyopadhyay05, 
Bandyopadhyay08}. $\rho_1(\lambda)$ and $\rho_2(\lambda)$ are 
given by

\begin{equation}\label{eq12}
\begin{array}{c}
\rho_1(\lambda) = \lambda|\psi_i\rangle\langle\psi_i| + (1-\lambda)\rho,
\\[1 ex]
\rho_2(\lambda) = \lambda( I/9) + (1-\lambda)\rho.
\end{array}
\end{equation} 
where $|\psi_i\rangle$ can be any state of the UPB given in 
Eq.~(\ref{eq9}). Both classes of states given above satisfy the range criterion 
\cite{Bandyopadhyay05, Bandyopadhyay08}. Clearly, $\rho_1(\lambda)$ is of 
rank-5 while $\rho_2(\lambda)$ is of full rank. Now, we compute the minimum 
eigenvalues of the operators  $(I\otimes\Lambda)(I\otimes u)\rho_1(\lambda)
(I\otimes u)^\dagger$ and $(I\otimes\Lambda)(I\otimes u)\rho_2 (\lambda)
(I\otimes u)^\dagger$. It is found that for a small range of $\lambda$ (compared 
to the original values as given in \cite{Bandyopadhyay05, Bandyopadhyay08}), it 
is possible to prove the inseparability of the states $\rho_i(\lambda);~i=1,2$ by 
applying $\Lambda_{u}$ (see Figure~\ref{fig1} and Figure~\ref{fig2}). The state 
(within the $3\otimes 3$ UPB) which is considered for $\rho_1(\lambda)$, is 
$|\psi_1\rangle$. 

\begin{figure}[h]
\centering
\includegraphics[scale=0.45]{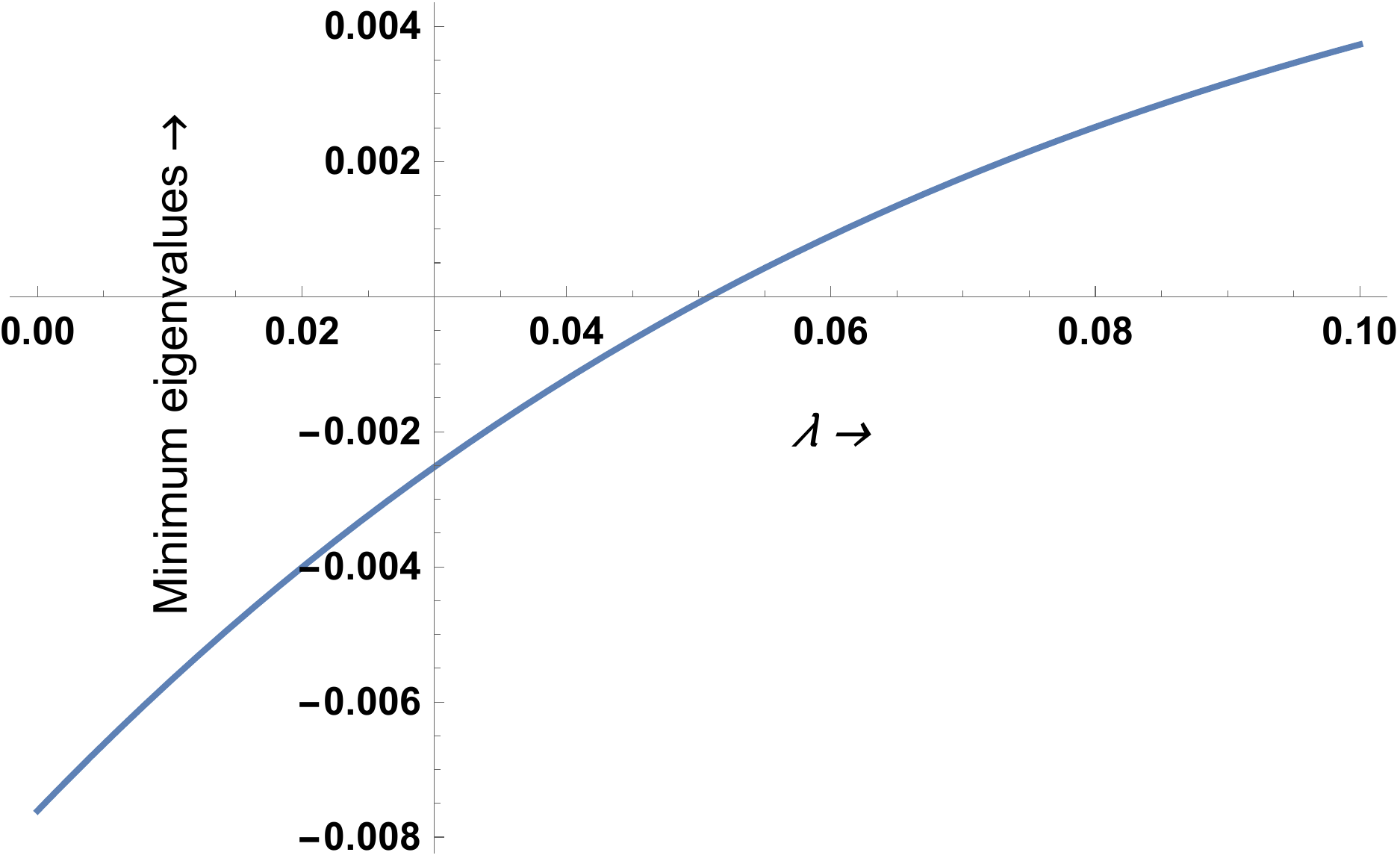}
\caption{Minimum eigenvalues of the operators $(I\otimes\Lambda)(I\otimes u)
\rho_1(\lambda)(I\otimes u)^\dagger$ are plotted for different values of $\lambda$. 
In the above figure it is clearly shown that for certain nonzero values of $\lambda$, 
it is possible to get negative eigenvalues of the operators $(I\otimes\Lambda)
(I\otimes u)\rho_1(\lambda)(I\otimes u)^\dagger$, resulting the detection of a 
subset of states $\rho_1(\lambda)$.}\label{fig1}
\end{figure}

\begin{figure}[h]
\centering
\includegraphics[scale=0.45]{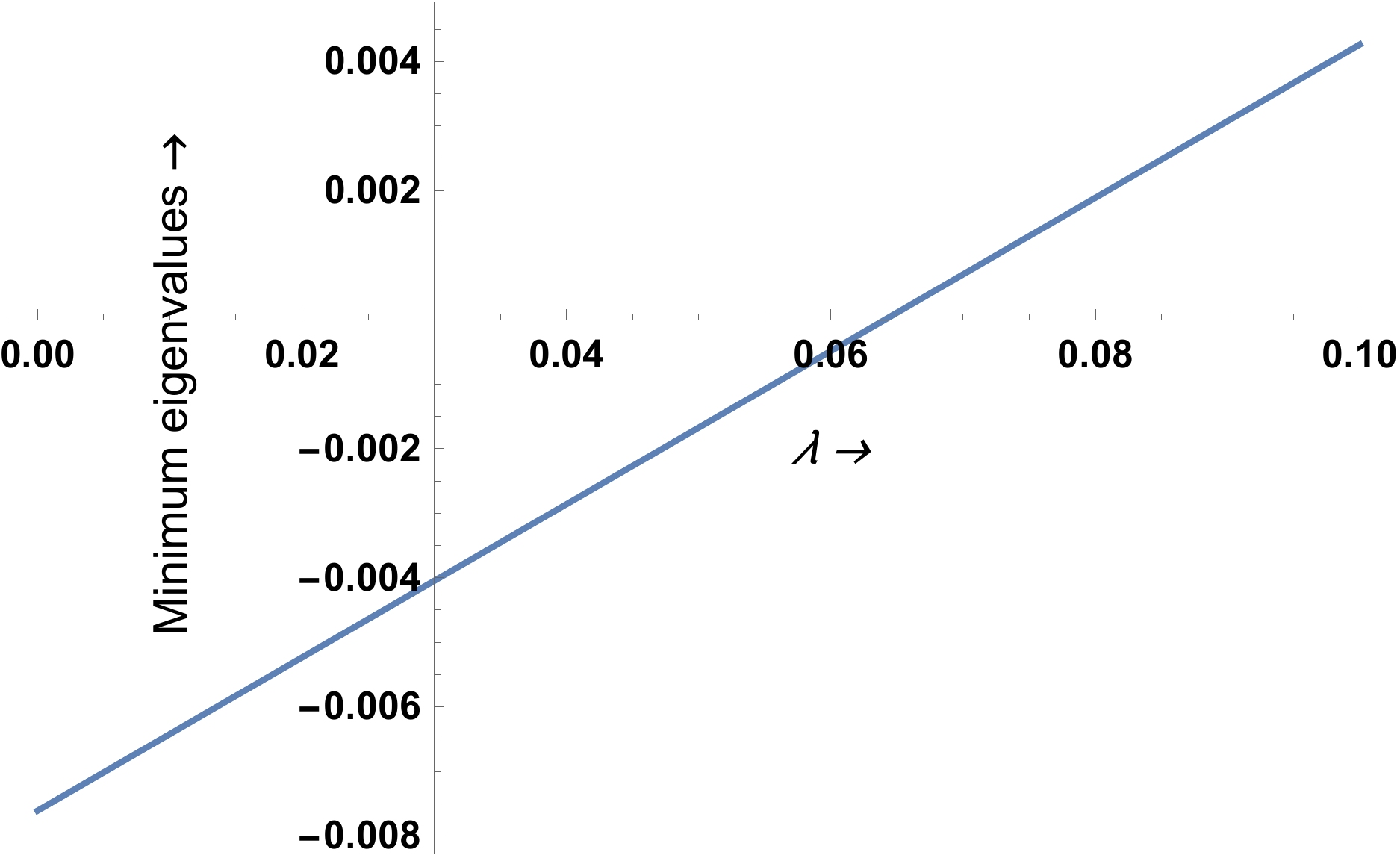}
\caption{Minimum eigenvalues of the operators $(I\otimes\Lambda)(I\otimes u)
\rho_2(\lambda)(I\otimes u)^\dagger$ are plotted for different values of $\lambda$. 
In the above figure it is clearly shown that for certain nonzero values of $\lambda$, 
it is possible to get negative eigenvalues of the operators $(I\otimes\Lambda)
(I\otimes u)\rho_2(\lambda)(I\otimes u)^\dagger$, resulting the detection of a 
subset of states $\rho_2(\lambda)$.}\label{fig2}
\end{figure}

Starting from the class of states $\rho_1(\lambda)$ in $3\otimes3$, a 
systematic method was developed in Ref.~\cite{Bandyopadhyay05} to 
produce such states in $d\otimes d$, having rank $r$ where $(d^2-4) 
\leq r \leq d^2$ . This was done by considering UPBs with real 
coefficients, that is, the states within a UPB are the normalized vectors 
of a real vector space. Here we say these UPBs as real UPBs (also see 
Ref.~\cite{Bandyopadhyay05}). To prove the inseparability of those states, 
the witness operator given in Eq.~(\ref{eq3}) was employed. However, 
in $d_1\otimes d_2$, we construct low-rank bound entangled states which 
satisfy the range criterion. 

In this regard, it is essential to mention the following: Given a set $S$ 
of four or lesser number of bipartite pure orthogonal product states in 
$d_1\otimes d_2$ then, the set S is extendible to a full basis in $d_1
\otimes d_2$ \cite{DiVincenzo03}. This particular fact guarantees that 
the bound entangled states of Ref.~\cite{Bandyopadhyay05} having 
rank $r$; $d^2-4\leq r \leq d^2$ in $d\otimes d$ must satisfy the range 
criterion. Therefore, it is highly important to construct bound entangled 
states that satisfy the range criterion and also having rank $r$, where 
$r<d_1d_2-4$ in $d_1\otimes d_2$; $d_1,d_1\geq3$. We are now ready 
to give a systematic protocol to do so: (i) Consider the class of UPBs for 
which if the {\it stopper} is removed from the UPB then the rest is 
extendible to a full basis, e.g., UPBs which are given in 
Refs.~\cite{Bennett99, DiVincenzo03, Halder18}. (ii) Following these 
constructions, it is possible to construct real UPBs of the above kind. 
(iii) The entanglement of an edge state corresponding to a UPB has 
robustness, i.e., if a product state is picked from a given UPB and is 
mixed with corresponding edge state then the resulting PPT state can 
be entangled. (iv) So, if the stopper is chosen from a real UPB of the 
above kind and mixed with the edge state with certain proportion to 
produce new bound entangled states then they must satisfy the range 
criterion. 

\paragraph*{Example}
We consider a real UPB in $4\otimes3$. For the construction, one can 
go through GenTiles2 UPBs of Ref.~\cite{DiVincenzo03}. The UPB and 
corresponding PPT entangled state (the edge state) are given as the 
following:

\begin{equation}\label{eq13}
\begin{array}{c}
|\phi_1\rangle = \frac{1}{\sqrt{2}}|1\rangle|1-2\rangle,~~
|\phi_2\rangle = \frac{1}{\sqrt{2}}|2\rangle|2-3\rangle,~~
|\phi_3\rangle = \frac{1}{\sqrt{2}}|3\rangle|3-1\rangle,\\
|\phi_4\rangle = \frac{1}{\sqrt{2}}|2-4\rangle|1\rangle,~~
|\phi_5\rangle = \frac{1}{\sqrt{2}}|3-4\rangle|2\rangle,~~
|\phi_6\rangle = \frac{1}{\sqrt{2}}|1-4\rangle|3\rangle,\\
|\phi_7\rangle = \frac{1}{2\sqrt{3}}|1+2+3+4\rangle|1+2+3\rangle,~~
\sigma = \frac{1}{5}\left(I-\sum_{i=1}^7|\phi_i\rangle\langle\phi_i|\right).
\end{array}
\end{equation}
The above edge state is invariant under partial transpose as it is 
produced due to a real UPB. Again, this state is of rank-5. We now 
consider the following class of states, given by
\begin{equation}\label{eq14}
\sigma_1(\lambda) = \lambda |\phi_7\rangle\langle\phi_7| + 
(1-\lambda)\sigma.
\end{equation}
The above states are of rank-6 $<(d_1d_2-4)$, again, the ranges of these 
states are the same and is spanned by the following product states:
\begin{equation}\label{eq15}
\begin{array}{c}
|\phi_1^\prime\rangle = \frac{1}{\sqrt{2}}|1\rangle|1+2\rangle,~~
|\phi_2^\prime\rangle = \frac{1}{\sqrt{2}}|2\rangle|2+3\rangle,~~
|\phi_3^\prime\rangle = \frac{1}{\sqrt{2}}|3\rangle|3+1\rangle,\\
|\phi_4^\prime\rangle = \frac{1}{\sqrt{2}}|2+4\rangle|1\rangle,~~
|\phi_5^\prime\rangle = \frac{1}{\sqrt{2}}|3+4\rangle|2\rangle,~~
|\phi_6^\prime\rangle = \frac{1}{\sqrt{2}}|1+4\rangle|3\rangle.
\end{array}
\end{equation}
We prove the inseparability of a subset of the states $\sigma_1(\lambda)$ 
by applying the same technique as employed to prove the inseparability of 
a subset of states $\rho_i(\lambda)$; $i=1,2$. We basically compute the 
negative eigenvalues for a range of $\lambda$ (see Figure~\ref{fig3}).

\begin{figure}[h]
\centering
\includegraphics[scale=0.45]{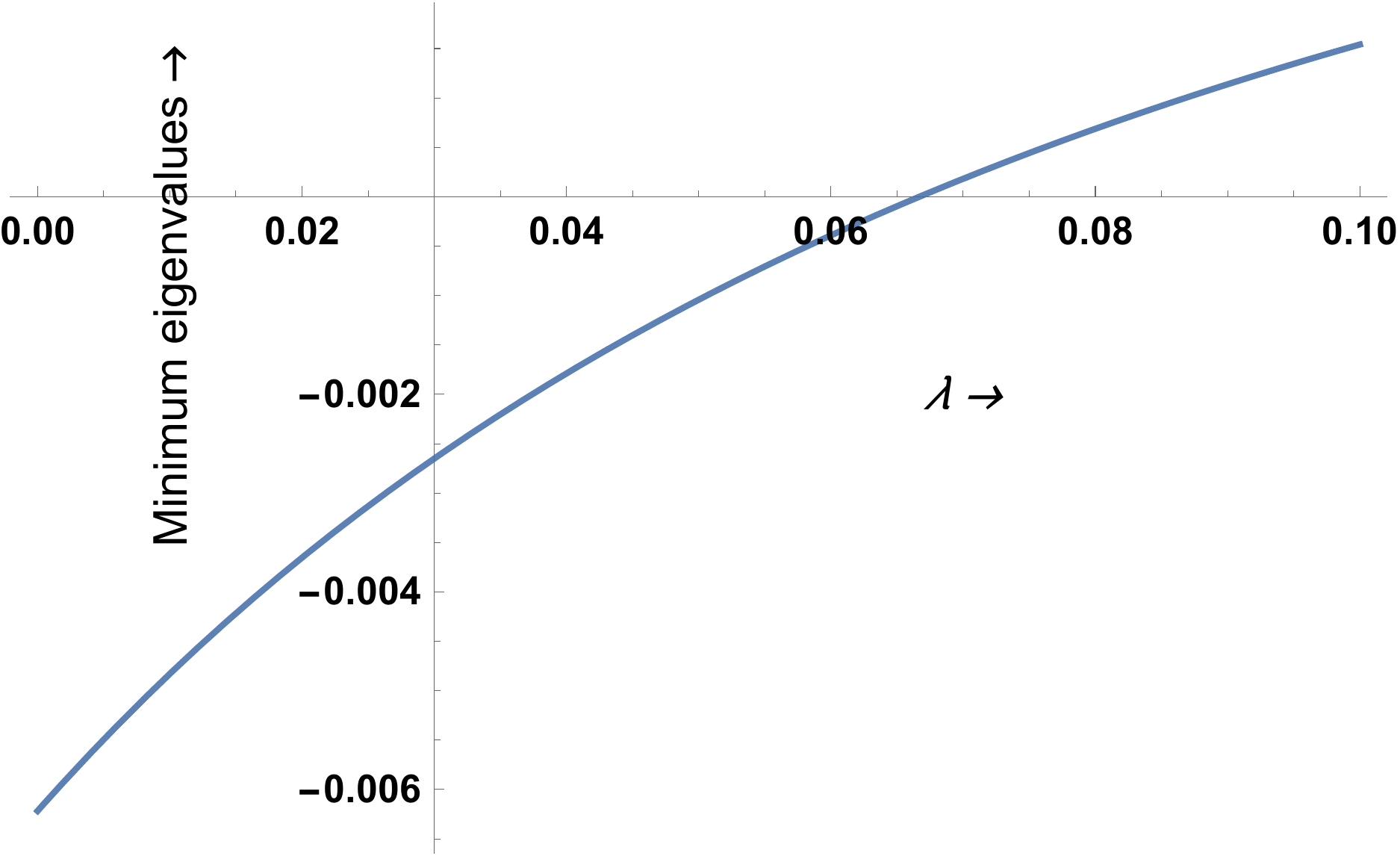}
\caption{Minimum eigenvalues of the operators $(I\otimes\Lambda)(I\otimes u)
\sigma_1(\lambda)(I\otimes u)^\dagger$ are plotted for different values of $\lambda$. 
In the above figure it is clearly shown that for certain nonzero values of $\lambda$, 
it is possible to get negative eigenvalues of the operators $(I\otimes\Lambda)
(I\otimes u)\sigma_1(\lambda)(I\otimes u)^\dagger$, resulting the detection of a 
subset of states $\sigma_1(\lambda)$.}\label{fig3}
\end{figure}

In general, to prove the existence of such class of states in $d_1\otimes d_2$ 
(which can be constructed from real UPBs), suitable witness operators can be 
employed. Now, we present the following theorem for any real UPB with the 
property that if the stopper is removed from the UPB then the rest is extendible 
to a full basis. These UPBs can be of any arbitrary cardinality (number of states 
present within a UPB).

\begin{theorem}\label{th1}
Consider any UPB (of the above kind) with cardinality $N$ in $d_1\otimes d_2$. 
Starting from such a UPB, it is possible to construct bound entangled states of 
ranks $(d_1d_2-N)+1$ to $d_1d_2$ in $d_1\otimes d_2$ with the property that 
they satisfy the range criterion.
\end{theorem}

\begin{proof}
Consider a set $S$ of pure orthogonal product states $\{|\psi_i\rangle
\}_{i=1}^N$. The set $S$ forms a real UPB in $d_1\otimes d_2$ with an 
additional property that if the stopper is removed from the set $S$ then 
the rest is extendible to a full basis. This property is important to produce 
the desired entangled states. The bound entangled state (the edge state) 
due to this UPB is given by
\begin{equation}\label{eq16}
\sigma_2 = \frac{1}{d_1d_2-N}\left(I-\sum_{i=1}^N|\psi_i\rangle\langle
\psi_i|\right).
\end{equation}
Now, consider a subset $S^\prime\subseteq S$ such that $S^\prime$ must 
include the stopper. Assume that $\mathbb{P}$ be the normalized projector 
onto the subspace spanned by the product states of $S^\prime$. Let us now 
consider the following class of states:
\begin{equation}\label{eq17}
\sigma_2(\lambda) = \lambda\mathbb{P} + (1-\lambda)\sigma_2.
\end{equation}
Notice that all of the above states must satisfy the range criterion. This is 
because of the following facts: (a) The above states are invariant under 
partial transpose. (b) The ranges of the above states are spanned by a set 
of orthogonal product states which are normalized vectors of a real vector 
space (this happens as $S^\prime$ contains the stopper). Now, to prove 
the inseparability of a subset of states from the above density matrices, we 
consider the same technique as given in Ref.~\cite{Bandyopadhyay05}. We 
consider the entanglement witness operator $\mathcal{W}$ as defined in 
Sec.~\ref{sec2}. Considering the trace $\mbox{Tr}[\mathcal{W}
\sigma_2(\lambda)]$, we obtain 
\begin{equation}\label{eq18}
\mbox{Tr}[\mathcal{W}\sigma_2(\lambda)] = (\lambda -\gamma).
\end{equation}
This quantity is less than zero when $0< \lambda < \gamma$, resulting 
the detection of PPT entangled states. Notice that if the subset $S^\prime$ 
contains only the stopper then the states $\sigma_2(\lambda)$ have the 
rank $(d_1d_2-N)+1$. Starting from this if the subset $S^\prime$ becomes 
exactly the same as $S$ then the states $\sigma_2(\lambda)$ have the rank 
$d_1d_2$ (full rank). Here the proof completes.
\end{proof}

A fundamental property of the class of states discussed in Theorem \ref{th1} 
is that the range of such a state is spanned by orthogonal product states. 
Therefore, it is quite natural to ask about the construction of bound entangled 
states which satisfy the range criterion but the range of such a state is not 
spanned by orthogonal product states, that is, the range of such a state is 
spanned by nonorthogonal product states. To answer this question, we 
consider a set $S$ of pure orthogonal product states in a real Hilbert space 
$\mathcal{H}$. Assume that the states of $S$ form a UCPB and they span 
the subspace $\mathcal{H}_S$ of $\mathcal{H}$. If the normalized projector 
onto the complementary subspace $\mathcal{H}_S^\perp$ is separable 
then that state must be written as the convex combination of nonorthogonal 
product states. Explicit construction of such a set of product states is given in 
\cite{DiVincenzo03}. Now, we present the following theorem.

\begin{theorem}\label{th2}
The UCPBs of the above kind are useful to construct bound entangled states 
which satisfy the range criterion but the range of which is not spanned by 
orthogonal product states. 
\end{theorem}

\begin{proof}
Let $S_1$ = $\{|\psi_i\rangle\}_{i=1}^{n^\prime}$ be such a UCPB in a 
Hilbert space $\mathcal{H}$. If these states span the subspace $\mathcal{H}
_{S_1}$ of $\mathcal{H}$ then the complementary subspace $\mathcal{H}
_{S_1}^\perp$ contains the states $\{|\psi_i\rangle\}_{i=n^\prime+1}^{n}$, 
where $n$ is strictly less then the net dimension of $\mathcal{H}$. The 
normalized projector $\sigma_3$ onto the subspace $\mathcal{H}_{S_1}^
\perp$ is separable. Therefore, the state $\sigma_3$ satisfies the range 
criterion. Notices that the states of $S_1$ and the states $\{|\psi_i\rangle\}_
{i=n^\prime+1}^{n}$ together form a UPB. We assume that the edge 
state corresponding to that UPB is $\sigma_3^\prime$. So, $\sigma_3$ can 
be written as the convex combination of $\sigma_3^\prime$ and some 
separable state $\delta_1$, for clarity see Eq.~(\ref{eq2}). We now define a 
class of partial transpose invariant states $\sigma_3(\lambda)$ as the following:
\begin{equation}\label{eq19}
\sigma_3(\lambda) = \lambda \delta_1 + (1-\lambda)\sigma_3^\prime.
\end{equation}
Clearly, for nonzero $\lambda$, all of them have the same range. So, the 
states $\sigma_3(\lambda)$ satisfy the range criterion. Next, to prove the 
inseparability of a subset of such states, one can follow the technique given 
in the proof of Theorem \ref{th1}. So, using the witness operator, given 
in Eq.~(\ref{eq3}), it is possible to have inseparable PPT states when 
$0<\lambda<\gamma$. So, these result in bound entangled states which 
satisfy the range criterion but the range of such a state is not spanned by 
orthogonal product states.
\end{proof}

The states constructed so far, are bipartite noisy bound entangled states 
which satisfy the range criterion. However, there also exist bipartite noisy 
bound entangled states which do not satisfy the range criterion. Moreover, 
it is possible to construct a class of bound entangled states which show 
maximal robustness of entanglement of an edge state. To realize this fact, 
consider a UCPB in a Hilbert space $\mathcal{H}$. This UCPB should be 
different compared to the above one in a sense that if the states of the 
UCPB span the subspace $\mathcal{H}_S$ of $\mathcal{H}$ then the 
complementary subspace $\mathcal{H}_S^\perp$ has product state deficit. 
Thus, the normalized projector onto the complementary subspace must be 
entangled and violet the range criterion. In fact consider the following 
class of states (having exactly the same range as that of the normalized 
projector onto the complementary subspace $\mathcal{H}_S^\perp$)

\begin{equation}\label{eq20}
\sigma_4(\lambda) = \lambda\delta_2 + (1-\lambda)\sigma_4^\prime,
\end{equation}
where $\delta_2$ is a separable state and $\sigma_4^\prime$ is an edge 
state. Note that due to product state deficit for all nonzero values of 
$\lambda$, the states $\sigma_4(\lambda)$ violate the range criterion. So, 
all of these states are entangled and due to the construction they must be 
PPT. These states also suggest that with any proportion the separable 
states $\delta_2$ is mixed with the edge state $\sigma_4^\prime$, the 
resulting states are inseparable. In this sense, the above states are showing 
maximal robustness of entanglement of the edge state. Interestingly, the 
witness operator given in Sec.~\ref{sec2} is not able to detect all the states 
$\sigma_4(\lambda)$ and it happens when $\lambda>\gamma$.

\paragraph*{Example:}
It is possible to construct a simple example of such a class of bound 
entangled states. Consider the UPB in $3\otimes3$ and the state $\rho$ 
as given earlier in Eq.~(\ref{eq9}). This UPB can be extended trivially 
to a $4\otimes3$ UPB by adding some product states $\{|41\rangle,~
|42\rangle,~|43\rangle\}$. Here, the notation $|ab\rangle$ stands for 
$|a\rangle\otimes|b\rangle$. Now, consider the following class of PPT 
states:

\begin{equation}\label{eq21}
\rho_3(\lambda) = \lambda|ab\rangle\langle ab| + (1-\lambda)\rho,
\end{equation}
where $|ab\rangle$ is any product state picked from the set $\{|41\rangle,
~|42\rangle,~|43\rangle\}$. Notice that the state $|ab\rangle$ neither 
belong to the $3\otimes3$ subspace where the pure states of 
Eq.~(\ref{eq9}) reside nor belong to the support of $\rho$. So, the 
five-dimensional subspace where the states $\rho_3(\lambda)$ are supported, 
has product state deficit. This certifies the violation of the range criterion 
by the states $\rho_3(\lambda)$ and therefore, we identify a different class 
of bipartite bound entangled states which are noisy and violet the range 
criterion.

There are other ways to construct noisy bound entangled states. In this 
regard we consider a different type of UPBs -- the states of which cannot 
be perfectly distinguished by separable measurements. Such UPBs can be 
found in Refs.~\cite{Bandyopadhyay15, Yang15}. We now present the 
following proposition and as a useful byproduct of the following proposition 
we obtain a new method to generate noisy bound entangled states.  

\begin{proposition}\label{prop1}
Consider a UPB, the states of which cannot be perfectly distinguished 
by separable measurements. From such a UPB if any state is missing 
then the resulting set becomes a UCPB. 
\end{proposition}

\begin{proof}
If there are $N$ states which form a UPB of the above kind then take 
any $N-1$ states without loss of generality. It is possible to construct 
rank-1 separable operators ($\Pi_i$) corresponding to these states. To 
complete a separable measurement the operator $(I-\sum_{i=1}^{N-1}
\Pi_i)$ must be separable. But this contradicts the fact that the states of 
the UPB cannot be distinguished by any separable measurement. Thus, 
the operator $(I-\sum_{i=1}^{N-1}\Pi_i)$ must be inseparable and as 
it is a projector, the inseparability results the fact that there must not 
be sufficient orthogonal product states present in the range of the 
operator $(I-\sum_{i=1}^{N-1}\Pi_i)$. Because if there are sufficient 
pure orthogonal product states then the operator must be separable.
Clearly, in a given Hilbert space $\mathcal{H}$, if the subset of any $N-1$ 
states of the UPB span the subspace $\mathcal{H}_{S^\prime}$ of 
$\mathcal{H}$ then $\mathcal{H}_{S^\prime}^\perp$ contains fewer 
orthogonal product states then its dimension. So, the subset is a UCPB. 
\end{proof}

Notice that the operators $(I-\sum_{i=1}^{N-1}\Pi_i)$ in the normalized 
form are noisy bound entangled states as the subset of any $N-1$ states 
is a UCPB. However, it is not known whether these states satisfy the range 
criterion or not. From the discussion so far, it is clear that if a state is missing 
from the UPB then it may result in a UCPB which can lead to the generation 
of noisy bound entangled state. But this may depend on which state is 
missing. However, there are scenarios when it really does not matter which 
state is missing as the resulting subset is always a UCPB (Proposition 
\ref{prop1}). Along with this line an interesting observation can be given in 
the following way: Consider the UPB, given in Eq.~(\ref{eq13}). Suppose, 
from this UPB the first state $|\phi_1\rangle$ is missing. Then the the rest 
product states spans a six-dimensional subspace $\mathcal{H}_{S^\prime}$. 
Applying the map $\Lambda_u$ of Eq.~(\ref{eq10}), it can be shown that 
the normalized projector onto the subspace $\mathcal{H}_{S^\prime}^\perp$ 
is inseparable. Thus, if the state $|\phi_1\rangle$ is missing from the UPB of 
Eq.~(\ref{eq13}) then the rest product states result in a UCPB. On the other 
hand, if the state $|\phi_7\rangle$ is missing from the same UPB then the 
rest product states can be extended to a full basis. 

\section{Conclusion}\label{sec4}
In this work, we have mainly focused on the construction of the noisy 
bound entangled states. Such a state can be written as a convex 
combination of an edge state and a separable state. Undoubtedly, 
such states shed light on the robustness of entanglement of edge states. 
Moreover, in a practical scenario it is never possible to eliminate noise 
completely and hence studying noisy entangled states have practical 
relevance. We have constructed an explicit protocol to produce 
low-rank noisy bound entangled states which satisfy the range criterion. 
In particular, the ranges of these states are spanned by orthogonal 
product states. Furthermore, we have discussed about ways to construct 
noisy bound entangled states from UCPBs. We have shown that the 
range of a UCPB generated bound entangled state may not be spanned 
by orthogonal product state still it can satisfy the range criterion.
For further studies one may consider the present problem of constructing 
distinct classes of noisy bound entangled states but without invoking 
product bases. It will also be interesting to examine whether these 
states will satisfy range criterion or not.

\section*{Acknowledgment}
\noindent R. S. acknowledges funding from SERB MATRICS MTR/2017/000431.

\bibliographystyle{elsarticle-num-names}
\biboptions{sort&compress}
\bibliography{ref}

\begin{thebibliography}{32}
\expandafter\ifx\csname natexlab\endcsname\relax\def\natexlab#1{#1}\fi
\providecommand{\url}[1]{\texttt{#1}}
\providecommand{\href}[2]{#2}
\providecommand{\path}[1]{#1}
\providecommand{\DOIprefix}{doi:}
\providecommand{\ArXivprefix}{arXiv:}
\providecommand{\URLprefix}{URL: }
\providecommand{\Pubmedprefix}{pmid:}
\providecommand{\doi}[1]{\href{http://dx.doi.org/#1}{\path{#1}}}
\providecommand{\Pubmed}[1]{\href{pmid:#1}{\path{#1}}}
\providecommand{\bibinfo}[2]{#2}
\ifx\xfnm\relax \def\xfnm[#1]{\unskip,\space#1}\fi
%Type = Article
\bibitem[{Horodecki et~al.(2009)Horodecki, Horodecki, Horodecki, and
  Horodecki}]{Horodecki09}
\bibinfo{author}{R.~Horodecki}, \bibinfo{author}{P.~Horodecki},
  \bibinfo{author}{M.~Horodecki}, \bibinfo{author}{K.~Horodecki},
\newblock \bibinfo{title}{Quantum entanglement},
\newblock \bibinfo{journal}{Rev. Mod. Phys.} \bibinfo{volume}{81}
  (\bibinfo{year}{2009}) \bibinfo{pages}{865--942}.
  \DOIprefix\doi{10.1103/RevModPhys.81.865}.
%Type = Article
\bibitem[{G{\"u}hne and T{\'o}th(2009)}]{Gune09}
\bibinfo{author}{O.~G{\"u}hne}, \bibinfo{author}{G.~T{\'o}th},
\newblock \bibinfo{title}{Entanglement detection},
\newblock \bibinfo{journal}{Phys. Rep.} \bibinfo{volume}{474}
  (\bibinfo{year}{2009}) \bibinfo{pages}{1--75}.
  \DOIprefix\doi{10.1016/j.physrep.2009.02.004}.
%Type = Article
\bibitem[{Horodecki(1997)}]{Horodecki97}
\bibinfo{author}{P.~Horodecki},
\newblock \bibinfo{title}{Separability criterion and inseparable mixed states
  with positive partial transposition},
\newblock \bibinfo{journal}{Phys. Lett. A} \bibinfo{volume}{232}
  (\bibinfo{year}{1997}) \bibinfo{pages}{333--339}.
  \DOIprefix\doi{10.1016/S0375-9601(97)00416-7}.
%Type = Article
\bibitem[{Horodecki et~al.(1998)Horodecki, Horodecki, and
  Horodecki}]{Horodecki98}
\bibinfo{author}{M.~Horodecki}, \bibinfo{author}{P.~Horodecki},
  \bibinfo{author}{R.~Horodecki},
\newblock \bibinfo{title}{Mixed-state entanglement and distillation: Is there a
  ``bound'' entanglement in nature?},
\newblock \bibinfo{journal}{Phys. Rev. Lett.} \bibinfo{volume}{80}
  (\bibinfo{year}{1998}) \bibinfo{pages}{5239--5242}.
  \DOIprefix\doi{10.1103/PhysRevLett.80.5239}.
%Type = Article
\bibitem[{Peres(1996)}]{Peres96}
\bibinfo{author}{A.~Peres},
\newblock \bibinfo{title}{Separability criterion for density matrices},
\newblock \bibinfo{journal}{Phys. Rev. Lett.} \bibinfo{volume}{77}
  (\bibinfo{year}{1996}) \bibinfo{pages}{1413--1415}.
  \DOIprefix\doi{10.1103/PhysRevLett.77.1413}.
%Type = Article
\bibitem[{Horodecki et~al.(1996)Horodecki, Horodecki, and
  Horodecki}]{Horodecki96}
\bibinfo{author}{M.~Horodecki}, \bibinfo{author}{P.~Horodecki},
  \bibinfo{author}{R.~Horodecki},
\newblock \bibinfo{title}{Separability of mixed states: necessary and
  sufficient conditions},
\newblock \bibinfo{journal}{Phys. Lett. A} \bibinfo{volume}{223}
  (\bibinfo{year}{1996}) \bibinfo{pages}{1--8}.
  \DOIprefix\doi{10.1016/S0375-9601(96)00706-2}.
%Type = Article
\bibitem[{Gurvits(2004)}]{Gurvits04}
\bibinfo{author}{L.~Gurvits},
\newblock \bibinfo{title}{Classical complexity and quantum entanglement},
\newblock \bibinfo{journal}{Journal of Computer and System Sciences}
  \bibinfo{volume}{69} (\bibinfo{year}{2004}) \bibinfo{pages}{448 -- 484}.
  \DOIprefix\doi{10.1016/j.jcss.2004.06.003}.
%Type = Article
\bibitem[{DiVincenzo et~al.(2000)DiVincenzo, Shor, Smolin, Terhal, and
  Thapliyal}]{DiVincenzo00}
\bibinfo{author}{D.~P. DiVincenzo}, \bibinfo{author}{P.~W. Shor},
  \bibinfo{author}{J.~A. Smolin}, \bibinfo{author}{B.~M. Terhal},
  \bibinfo{author}{A.~V. Thapliyal},
\newblock \bibinfo{title}{Evidence for bound entangled states with negative
  partial transpose},
\newblock \bibinfo{journal}{Phys. Rev. A} \bibinfo{volume}{61}
  (\bibinfo{year}{2000}) \bibinfo{pages}{062312}.
  \DOIprefix\doi{10.1103/PhysRevA.61.062312}.
%Type = Article
\bibitem[{D\"ur et~al.(2000)D\"ur, Cirac, Lewenstein, and Bru\ss{}}]{Dur00}
\bibinfo{author}{W.~D\"ur}, \bibinfo{author}{J.~I. Cirac},
  \bibinfo{author}{M.~Lewenstein}, \bibinfo{author}{D.~Bru\ss{}},
\newblock \bibinfo{title}{Distillability and partial transposition in bipartite
  systems},
\newblock \bibinfo{journal}{Phys. Rev. A} \bibinfo{volume}{61}
  (\bibinfo{year}{2000}) \bibinfo{pages}{062313}.
  \DOIprefix\doi{10.1103/PhysRevA.61.062313}.
%Type = Article
\bibitem[{Pankowski et~al.(2010)Pankowski, Piani, Horodecki, and
  Horodecki}]{Pankowski10}
\bibinfo{author}{{\L}.~Pankowski}, \bibinfo{author}{M.~Piani},
  \bibinfo{author}{M.~Horodecki}, \bibinfo{author}{P.~Horodecki},
\newblock \bibinfo{title}{A few steps more towards {NPT} bound entanglement},
\newblock \bibinfo{journal}{IEEE Trans. Inf. Theory} \bibinfo{volume}{56}
  (\bibinfo{year}{2010}) \bibinfo{pages}{4085--4100}.
  \DOIprefix\doi{10.1109/TIT.2010.2050810}.
%Type = Article
\bibitem[{Bandyopadhyay et~al.(2005)Bandyopadhyay, Ghosh, and
  Roychowdhury}]{Bandyopadhyay05}
\bibinfo{author}{S.~Bandyopadhyay}, \bibinfo{author}{S.~Ghosh},
  \bibinfo{author}{V.~Roychowdhury},
\newblock \bibinfo{title}{Non-full-rank bound entangled states satisfying the
  range criterion},
\newblock \bibinfo{journal}{Phys. Rev. A} \bibinfo{volume}{71}
  (\bibinfo{year}{2005}) \bibinfo{pages}{012316}.
  \DOIprefix\doi{10.1103/PhysRevA.71.012316}.
%Type = Article
\bibitem[{Bandyopadhyay et~al.(2008)Bandyopadhyay, Ghosh, and
  Roychowdhury}]{Bandyopadhyay08}
\bibinfo{author}{S.~Bandyopadhyay}, \bibinfo{author}{S.~Ghosh},
  \bibinfo{author}{V.~Roychowdhury},
\newblock \bibinfo{title}{Robustness of entangled states that are positive
  under partial transposition},
\newblock \bibinfo{journal}{Phys. Rev. A} \bibinfo{volume}{77}
  (\bibinfo{year}{2008}) \bibinfo{pages}{032318}.
  \DOIprefix\doi{10.1103/PhysRevA.77.032318}.
%Type = Article
\bibitem[{Bennett et~al.(1999)Bennett, DiVincenzo, Mor, Shor, Smolin, and
  Terhal}]{Bennett99}
\bibinfo{author}{C.~H. Bennett}, \bibinfo{author}{D.~P. DiVincenzo},
  \bibinfo{author}{T.~Mor}, \bibinfo{author}{P.~W. Shor},
  \bibinfo{author}{J.~A. Smolin}, \bibinfo{author}{B.~M. Terhal},
\newblock \bibinfo{title}{Unextendible product bases and bound entanglement},
\newblock \bibinfo{journal}{Phys. Rev. Lett.} \bibinfo{volume}{82}
  (\bibinfo{year}{1999}) \bibinfo{pages}{5385--5388}.
  \DOIprefix\doi{10.1103/PhysRevLett.82.5385}.
%Type = Article
\bibitem[{Lewenstein et~al.(2001)Lewenstein, Kraus, Horodecki, and
  Cirac}]{Lewenstein01}
\bibinfo{author}{M.~Lewenstein}, \bibinfo{author}{B.~Kraus},
  \bibinfo{author}{P.~Horodecki}, \bibinfo{author}{J.~I. Cirac},
\newblock \bibinfo{title}{Characterization of separable states and entanglement
  witnesses},
\newblock \bibinfo{journal}{Phys. Rev. A} \bibinfo{volume}{63}
  (\bibinfo{year}{2001}) \bibinfo{pages}{044304}.
  \DOIprefix\doi{10.1103/PhysRevA.63.044304}.
%Type = Article
\bibitem[{DiVincenzo et~al.(2003)DiVincenzo, Mor, Shor, Smolin, and
  Terhal}]{DiVincenzo03}
\bibinfo{author}{D.~P. DiVincenzo}, \bibinfo{author}{T.~Mor},
  \bibinfo{author}{P.~W. Shor}, \bibinfo{author}{J.~A. Smolin},
  \bibinfo{author}{B.~M. Terhal},
\newblock \bibinfo{title}{Unextendible product bases, uncompletable product
  bases and bound entanglement},
\newblock \bibinfo{journal}{Communications in Mathematical Physics}
  \bibinfo{volume}{238} (\bibinfo{year}{2003}) \bibinfo{pages}{379--410}.
  \DOIprefix\doi{10.1007/s00220-003-0877-6}.
%Type = Article
\bibitem[{Horodecki et~al.(2005)Horodecki, Horodecki, Horodecki, and
  Oppenheim}]{Horodecki05}
\bibinfo{author}{K.~Horodecki}, \bibinfo{author}{M.~Horodecki},
  \bibinfo{author}{P.~Horodecki}, \bibinfo{author}{J.~Oppenheim},
\newblock \bibinfo{title}{Secure key from bound entanglement},
\newblock \bibinfo{journal}{Phys. Rev. Lett.} \bibinfo{volume}{94}
  (\bibinfo{year}{2005}) \bibinfo{pages}{160502}.
  \DOIprefix\doi{10.1103/PhysRevLett.94.160502}.
%Type = Article
\bibitem[{Horodecki et~al.(2008)Horodecki, Pankowski, Horodecki, and
  Horodecki}]{Horodecki08}
\bibinfo{author}{K.~Horodecki}, \bibinfo{author}{{\L}.~Pankowski},
  \bibinfo{author}{M.~Horodecki}, \bibinfo{author}{P.~Horodecki},
\newblock \bibinfo{title}{Low-dimensional bound entanglement with one-way
  distillable cryptographic key},
\newblock \bibinfo{journal}{IEEE Trans. Inf. Theory} \bibinfo{volume}{54}
  (\bibinfo{year}{2008}) \bibinfo{pages}{2621--2625}.
  \DOIprefix\doi{10.1109/TIT.2008.921709}.
%Type = Article
\bibitem[{Horodecki et~al.(2009)Horodecki, Horodecki, Horodecki, and
  Oppenheim}]{Horodecki09-1}
\bibinfo{author}{K.~Horodecki}, \bibinfo{author}{M.~Horodecki},
  \bibinfo{author}{P.~Horodecki}, \bibinfo{author}{J.~Oppenheim},
\newblock \bibinfo{title}{General paradigm for distilling classical key from
  quantum states},
\newblock \bibinfo{journal}{IEEE Trans. Inf. Theory} \bibinfo{volume}{55}
  (\bibinfo{year}{2009}) \bibinfo{pages}{1898--1929}.
  \DOIprefix\doi{10.1109/TIT.2008.2009798}.
%Type = Article
\bibitem[{T\'{o}th and V\'{e}rtesi(2018)}]{Toth18}
\bibinfo{author}{G.~T\'{o}th}, \bibinfo{author}{T.~V\'{e}rtesi},
\newblock \bibinfo{title}{Quantum states with a positive partial transpose are
  useful for metrology},
\newblock \bibinfo{journal}{Phys. Rev. Lett.} \bibinfo{volume}{120}
  (\bibinfo{year}{2018}) \bibinfo{pages}{020506}.
  \DOIprefix\doi{10.1103/PhysRevLett.120.020506}.
%Type = Article
\bibitem[{Moroder et~al.(2014)Moroder, Gittsovich, Huber, and
  G\"uhne}]{Moroder14}
\bibinfo{author}{T.~Moroder}, \bibinfo{author}{O.~Gittsovich},
  \bibinfo{author}{M.~Huber}, \bibinfo{author}{O.~G\"uhne},
\newblock \bibinfo{title}{Steering bound entangled states: A counterexample to
  the stronger peres conjecture},
\newblock \bibinfo{journal}{Phys. Rev. Lett.} \bibinfo{volume}{113}
  (\bibinfo{year}{2014}) \bibinfo{pages}{050404}.
  \DOIprefix\doi{10.1103/PhysRevLett.113.050404}.
%Type = Article
\bibitem[{V\'{e}rtesi and Brunner(2014)}]{Vertesi14}
\bibinfo{author}{T.~V\'{e}rtesi}, \bibinfo{author}{N.~Brunner},
\newblock \bibinfo{title}{Disproving the peres conjecture by showing bell
  nonlocality from bound entanglement},
\newblock \bibinfo{journal}{Nature Communications} \bibinfo{volume}{5}
  (\bibinfo{year}{2014}) \bibinfo{pages}{5297}.
  \DOIprefix\doi{10.1038/ncomms6297}.
%Type = Article
\bibitem[{Yu and Oh(2017)}]{Yu17}
\bibinfo{author}{S.~Yu}, \bibinfo{author}{C.~H. Oh},
\newblock \bibinfo{title}{Family of nonlocal bound entangled states},
\newblock \bibinfo{journal}{Phys. Rev. A} \bibinfo{volume}{95}
  (\bibinfo{year}{2017}) \bibinfo{pages}{032111}.
  \DOIprefix\doi{10.1103/PhysRevA.95.032111}.
%Type = Article
\bibitem[{Sindici and Piani(2018)}]{Sindici18}
\bibinfo{author}{E.~Sindici}, \bibinfo{author}{M.~Piani},
\newblock \bibinfo{title}{Simple class of bound entangled states based on the
  properties of the antisymmetric subspace},
\newblock \bibinfo{journal}{Phys. Rev. A} \bibinfo{volume}{97}
  (\bibinfo{year}{2018}) \bibinfo{pages}{032319}.
  \DOIprefix\doi{10.1103/PhysRevA.97.032319}.
%Type = Article
\bibitem[{Sent\'{i}s et~al.(2018)Sent\'{i}s, Greiner, Shang, Siewert, and
  Kleinmann}]{Sentis18}
\bibinfo{author}{G.~Sent\'{i}s}, \bibinfo{author}{J.~N. Greiner},
  \bibinfo{author}{J.~Shang}, \bibinfo{author}{J.~Siewert},
  \bibinfo{author}{M.~Kleinmann},
\newblock \bibinfo{title}{Bound entangled states fit for robust experimental
  verification},
\newblock \bibinfo{journal}{Quantum} \bibinfo{volume}{2} (\bibinfo{year}{2018})
  \bibinfo{pages}{113}. \DOIprefix\doi{10.22331/q-2018-12-18-113}.
%Type = Article
\bibitem[{Choi(1975)}]{Choi75}
\bibinfo{author}{M.~D. Choi},
\newblock \bibinfo{title}{Positive semidefinite biquadratic forms},
\newblock \bibinfo{journal}{Linear Algebra and its Applications}
  \bibinfo{volume}{12} (\bibinfo{year}{1975}) \bibinfo{pages}{95--100}.
  \DOIprefix\doi{10.1016/0024-3795(75)90058-0}.
%Type = Article
\bibitem[{Chru\'{s}ci\'{n}ski and Sarbicki(2014)}]{Chruscinski14}
\bibinfo{author}{D.~Chru\'{s}ci\'{n}ski}, \bibinfo{author}{G.~Sarbicki},
\newblock \bibinfo{title}{Entanglement witnesses: construction, analysis and
  classification},
\newblock \bibinfo{journal}{Journal of Physics A: Mathematical and Theoretical}
  \bibinfo{volume}{47} (\bibinfo{year}{2014}) \bibinfo{pages}{483001}.
  \DOIprefix\doi{10.1088/1751-8113/47/48/483001}.
%Type = Article
\bibitem[{Terhal(2001)}]{Terhal01}
\bibinfo{author}{B.~M. Terhal},
\newblock \bibinfo{title}{A family of indecomposable positive linear maps based
  on entangled quantum states},
\newblock \bibinfo{journal}{Linear Algebra and its Applications}
  \bibinfo{volume}{323} (\bibinfo{year}{2001}) \bibinfo{pages}{61--73}.
  \DOIprefix\doi{10.1016/S0024-3795(00)00251-2}.
%Type = Article
\bibitem[{G\"uhne et~al.(2007)G\"uhne, Hyllus, Gittsovich, and
  Eisert}]{Guhne06}
\bibinfo{author}{O.~G\"uhne}, \bibinfo{author}{P.~Hyllus},
  \bibinfo{author}{O.~Gittsovich}, \bibinfo{author}{J.~Eisert},
\newblock \bibinfo{title}{Covariance matrices and the separability problem},
\newblock \bibinfo{journal}{Phys. Rev. Lett.} \bibinfo{volume}{99}
  (\bibinfo{year}{2007}) \bibinfo{pages}{130504}.
  \DOIprefix\doi{10.1103/PhysRevLett.99.130504}.
%Type = Article
\bibitem[{Sengupta and Arvind(2013)}]{Sengupta13}
\bibinfo{author}{R.~Sengupta}, \bibinfo{author}{Arvind},
\newblock \bibinfo{title}{Extremal extensions of entanglement witnesses and
  their connection with unextendable product bases},
\newblock \bibinfo{journal}{Phys. Rev. A} \bibinfo{volume}{87}
  (\bibinfo{year}{2013}) \bibinfo{pages}{012318}.
  \DOIprefix\doi{10.1103/PhysRevA.87.012318}.
%Type = Article
\bibitem[{Halder et~al.(2018)Halder, Banik, and Ghosh}]{Halder18}
\bibinfo{author}{S.~Halder}, \bibinfo{author}{M.~Banik},
  \bibinfo{author}{S.~Ghosh},
\newblock \bibinfo{title}{New family of bound entangled states residing on the
  boundary of peres set},
\newblock \bibinfo{journal}{arXiv:1801.00405 [quant-ph]}
  (\bibinfo{year}{2018}). \DOIprefix\doi{arxiv.org/abs/1801.00405}.
%Type = Article
\bibitem[{Bandyopadhyay et~al.(2015)Bandyopadhyay, Cosentino, Johnston, Russo,
  Watrous, and Yu}]{Bandyopadhyay15}
\bibinfo{author}{S.~Bandyopadhyay}, \bibinfo{author}{A.~Cosentino},
  \bibinfo{author}{N.~Johnston}, \bibinfo{author}{V.~Russo},
  \bibinfo{author}{J.~Watrous}, \bibinfo{author}{N.~Yu},
\newblock \bibinfo{title}{Limitations on separable measurements by convex
  optimization},
\newblock \bibinfo{journal}{IEEE Trans. Inf. Theory} \bibinfo{volume}{61}
  (\bibinfo{year}{2015}) \bibinfo{pages}{3593--3604}.
  \DOIprefix\doi{10.1109/TIT.2015.2417755}.
%Type = Article
\bibitem[{Yang et~al.(2015)Yang, Gao, Xu, Zuo, Zhang, and Wen}]{Yang15}
\bibinfo{author}{Y.-H. Yang}, \bibinfo{author}{F.~Gao}, \bibinfo{author}{G.-B.
  Xu}, \bibinfo{author}{H.-J. Zuo}, \bibinfo{author}{Z.-C. Zhang},
  \bibinfo{author}{Q.-Y. Wen},
\newblock \bibinfo{title}{Characterizing unextendible product bases in
  qutrit-ququad system},
\newblock \bibinfo{journal}{Scientific Reports} \bibinfo{volume}{5}
  (\bibinfo{year}{2015}) \bibinfo{pages}{11963}.
  \DOIprefix\doi{10.1038/srep11963}.

\end{thebibliography}
\end{document}